\documentclass[twocolumn]{autart}    
\usepackage{svg}
\usepackage{booktabs}
\usepackage{graphicx}          
\usepackage{amsmath}
\usepackage{amssymb}
\usepackage{algorithm}      
\usepackage{algpseudocode}  

\newtheorem{assumption}{Assumption}
\newtheorem{lemma}{Lemma}[section]
\newtheorem{theorem}[lemma]{Theorem}
\newtheorem{remark}[lemma]{Remark}

\makeatletter
\def\ps@copyright{%
	\let\@mkboth\@gobbletwo
	\def\@oddhead{}%
	\let\@evenhead\@oddhead
	\def\@oddfoot{\hfil\thepage\hfil}%
	\let\@evenfoot\@oddfoot
}
\makeatother

\begin{document}

\begin{frontmatter}

\title{PANDA: A Matrix-Free Differentiable NMPC Solver \\ via Proximal Averaged Quasi-Newton with \\ Adaptive Linesearch Algorithm} 
\vspace{-6mm}


\thanks[t2]{Corresponding author.}
\author[jlu]{Yuankun Chen}\ead{chenyk24@mails.jlu.edu.cn},
\author[jlu]{Zifei Nie\thanksref{t2}}\ead{zifei\_nie@jlu.edu.cn},
\author[jlu]{Xun Gong}\ead{xungong@jlu.edu.cn},
\author[jluc]{Yunfeng Hu}\ead{huyf@jlu.edu.cn},
\author[tj]{Hong Chen}\ead{chenhong2019@tongji.edu.cn},

\address[jlu]{School of Artificial Intelligence, Jilin University, Changchun, China.}  
\address[jluc]{Department of Control Science and Engineering, Jilin University, Changchun, China.}
\address[tj]{College of Electronics and Information Engineering, Tongji University, Shanghai, China.}
\vspace{-8mm}

\begin{abstract}                          
Differentiable nonlinear model predictive control (NMPC) provides a principled way to embed optimal control structure into end-to-end learning paradigms, but its practical use is often limited by the computational and memory costs of both forward optimization and backward sensitivity propagation. This brief proposes \texttt{PANDA}, a matrix-free solver for differentiable NMPC. In the forward pass, \texttt{PANDA} combines proximal-gradient iterations with quasi-Newton acceleration and introduces an adaptive stepsize enlargement mechanism to mitigate the conservativeness of monotone stepsize reduction. The resulting stepsize behavior and its effect on local convergence are theoretically analyzed. In the backward pass, \texttt{PANDA} performs implicit differentiation from the residual equation and computes adjoint sensitivities using Krylov-subspace iterative methods together with automatic-differentiation-based Matrix-Vector product operators, thereby avoiding explicit Hessian and Jacobian construction. The method is evaluated on a nonconvex trailer NMPC problem embedded in an imitation learning task. The results show that \texttt{PANDA} achieves much faster forward and backward computation and lower memory overhead than representative differentiable optimization solvers, while maintaining effective imitation learning performance.

\end{abstract}
\end{frontmatter}

\section{Introduction}

Differentiable optimization provides a principled framework for embedding
optimization-based planning and decision-making into learning pipeline as inductive biases. In particular, differentiable nonlinear model predictive control (NMPC), which combines the physical interpretability and safety guarantees of MPC with the adaptability of learning-based frameworks, has been widely applied in autonomous driving and robotics, across tasks including imitation learning \cite{chen2025gauss}, reinforcement learning \cite{adabag2025differentiable}, and end-to-end planning \cite{huang2023E2E}.
\vspace{-2mm}

However, the practical deployment of differentiable NMPC is still hindered by its high computational cost and memory consumption. In the \textbf{forward pass}, namely, the process of solving the underlying optimization problem to obtain the optimal solution, one needs to handle a nonlinear, constrained, and often nonconvex problem. In the \textbf{backward pass}, namely, the process of computing the derivatives of the optimal solution with respect to the learnable parameters, one typically needs to perform time-consuming operations such as solving linear systems. As a result, compared with standard neural network layers, efficient computation of differentiable NMPC layer remains a major challenge.
\vspace{-2mm}

For the \textbf{forward pass}, existing methods can be broadly categorized into first-order and second-order approaches. First-order methods are attractive for their lightweight computational structure, but they often exhibit limited convergence speed. Second-order methods usually achieve faster convergence and are therefore widely adopted in practical solvers such as CasADi\cite{andersson2018casadi} and acados\cite{verschueren2022acados}, although this often comes at the price of higher memory cost. To balance these trade-offs, \cite{stella2017simple,themelis2018forward} combines the low-cost structure of first-order iterations with second-order acceleration mechanisms, thereby achieving high computational efficiency. Building on this line of work, \cite{de2022proximal} further imposes stricter stepsize restrictions based on local Lipschitz properties. However, such stepsize rules remain monotone decreasing, which may lead to overly conservative solution strategies and potentially slow down convergence in practice.
\vspace{-2mm}

For the \textbf{backward pass}, most existing methods start from the Karush--Kuhn--Tucker (KKT) conditions and apply the implicit function theorem to derive a linearized sensitivity system. Based on this system, some approaches construct auxiliary systems to compute the required sensitivities~\cite{jin2021safe,amos2018differentiable}, while others accelerate the computation by exploiting specific structures~\cite{chen2025gauss}. Despite these differences, such methods typically rely on the explicit construction of Hessian or Jacobian matrices. As a result, their memory cost often scales quadratically with the problem size, which can become a major bottleneck for long-horizon optimal control problems and for deployment on memory-constrained platforms.
\vspace{-2mm}

To address these challenges, we propose \texttt{PANDA}, i.e., a \textbf{P}roximal \textbf{A}veraged quasi-\textbf{N}ewton with a\textbf{D}aptive Linesearch \textbf{A}lgorithm. In the \textbf{forward pass}, \texttt{PANDA} revisits local Lipschitz information and allows the admissible stepsize to increase, unlike existing monotone decreasing rules. In the \textbf{backward pass}, \texttt{PANDA} focuses on Matrix-Vector product operators rather than explicit matrices, reducing the computational and memory burden of sensitivity propagation.
The main contributions are summarized as follows.
\begin{itemize}
	\item We propose \texttt{PANDA}\footnote{Code available: {https://github.com/optiXlab1/PANDA}.}, a matrix-free solver for differentiable NMPC, which improves the efficiency of both forward optimization and backward sensitivity propagation within a unified framework.
	
	\item For the forward pass, we extend the proximal quasi-Newton framework with an adaptive stepsize enlargement mechanism that mitigates the conservativeness of monotone stepsize reduction. The mechanism and applicability of this strategy are analyzed from both theoretical and practical
	perspectives.
	
	\item For the backward pass, we derive a residual-based implicit differentiation scheme and compute the adjoint sensitivity in a matrix-free manner using Krylov-subspace iterative methods and automatic-differentiation-enabled Jacobian-Vector or Vector-Jacobian product operators.
	
	\item We evaluate the proposed \texttt{PANDA} on a nonconvex trailer NMPC problem embedded in an imitation learning task. The results demonstrate faster forward and backward computation and lower memory overhead than several popular baselines, while maintaining comparable imitation learning performance.
\end{itemize}
\vspace{-2mm}
\section{Problem Setting and Preliminaries}
\vspace{-2mm}
Consider the following parameterized finite-horizon optimal control problem:
\begin{equation}
	\scalebox{0.95}{$
		\begin{aligned}
			\min_{\{x_k,u_k\}} \quad
			& \sum_{k=0}^{N_p-1} \ell_k(x_k,u_k;\theta)
			+ g_k(u_k;\theta)
			+ \ell_{N_p}(x_{N_p};\theta) \\
			\text{s.t.}\quad
			& x_0 = \bar x, \\
			& x_{k+1} = f_k(x_k,u_k;\theta), \quad k=0,\dots,N_p-1,
		\end{aligned}
		$}
	\label{eq:ocp}
\end{equation}
where \(x_k\in\mathbb{R}^{n_x}\), \(u_k\in\mathbb{R}^{n_u}\), \(N_p\) is the prediction horizon, and \(\theta\in\mathbb{R}^p\) denotes a vector of tunable parameters. Here, \(\ell_k\) and \(\ell_{N_p}\) are smooth cost functions that may also include soft constraints through penalty terms, while \(g_k\) is a possibly nonsmooth term and can be used to encode hard constraints or other structured regularization.
\vspace{-2mm}

Problems of the form \eqref{eq:ocp} commonly arise in practical control applications with tunable design parameters. Beyond solving \eqref{eq:ocp} for a given parameter \(\theta\), one often seeks to adjust \(\theta\) so that the resulting optimal trajectory performs well with respect to a prescribed task level criterion. Let
$L:\mathbb{R}^{N_p n_u}\to\mathbb{R}$
be such a task level loss:
\begin{equation}
	\mathcal{L}(\theta)
	:=
	L\bigl(u_{0:N_p-1}^*(\theta)\bigr),
	\label{eq:outer-loss}
\end{equation}
where \(u^*(\theta)\) denotes a locally optimal solution of \eqref{eq:ocp}, and \(L\) may also include the state trajectory \(x^*(\theta)\), which is represented by \(\bar x\) and \(u^*(\theta)\) via system dynamics.
\subsection{Forward-Backward Splitting}
\vspace{-2mm}
To obtain a control-only formulation of \eqref{eq:ocp}, we condense the state
trajectory by forward simulation. Define the control stack:
\begin{equation}
	u:=\mathrm{col}(u_0,\dots,u_{N_p-1})\in\mathbb{R}^{N_p n_u},
\end{equation}
and let $\boldsymbol{F}_{k,\theta}(u)$ denote the state at stage $k$ obtained by forward simulation, namely
\begin{equation}
	\boldsymbol{F}_{0,\theta}(u):=\bar x(\theta), \;
	\boldsymbol{F}_{k+1,\theta}(u):=f_k(\boldsymbol{F}_{k,\theta}(u),u_k;\theta).
\end{equation}
The inner problem can then be rewritten in the composite form:
\begin{equation}
	\min_{u\in\mathbb{R}^{N_p n_u}}
	\;\varphi_\theta(u)
	:=
	\ell_\theta(u)+g_\theta(u),
	\label{eq:composite-problem}
\end{equation}
with
\begin{align*}
	\ell_\theta(u)
	&:=
	\sum_{k=0}^{N_p-1}\ell_k(\boldsymbol{F}_{k,\theta}(u),u_k;\theta)
	+\ell_{N_p}(\boldsymbol{F}_{N_p,\theta}(u);\theta), \\
	g_\theta(u)
	&:=
	\sum_{k=0}^{N_p-1} g_k(u_k;\theta).
\end{align*}
For simplicity, we suppress the dependence on \(\theta\) whenever no ambiguity arises, and write
\(\varphi\), \(\ell\), \(g\), \(T_\gamma\), \(R_\gamma\), \(\varphi_\gamma\) in place of
\(\varphi_\theta\), \(\ell_\theta\), \(g_\theta\), \(T_{\gamma,\theta}\), \(R_{\gamma,\theta}\), \(\varphi_{\gamma,\theta}\), respectively.
\vspace{-2mm}

For a proper function $g:\mathbb{R}^m\to\overline{\mathbb{R}}$ and a stepsize $\gamma>0$, the proximal mapping is
\begin{equation}
	\operatorname{prox}_{\gamma g}(u)
	:=
	\arg\min_{w\in\mathbb{R}^m}
	\left\{
	g(w)+\frac{1}{2\gamma}\|w-u\|^2
	\right\}.
	\label{eq:prox-def}
\end{equation}
When $g$ is nonconvex, this mapping may be set-valued. We say that $g$ is prox-bounded if there exists
$\gamma_g\in(0,+\infty]$ such that $\operatorname{prox}_{\gamma g}(v)\neq\emptyset$ for all
$v\in\mathbb{R}^m$ and all $\gamma\in(0,\gamma_g)$; the value $\gamma_g$ is called the
prox-boundedness threshold.
\vspace{-2mm}
\begin{assumption}[Basic assumptions]
	\label{ass:basic}
	For the composite problem \eqref{eq:composite-problem}, and for each parameter value $\theta$ of interest, the following hold:
	\vspace{-2mm}
	\begin{enumerate}
		\item[\textup{\textbf{A1}}] $\ell(u)$ is continuously differentiable and $\nabla_u \ell$ is locally Lipschitz continuous.
		\item[\textup{\textbf{A2}}] $g(u)$ is proper, closed, and $\gamma_g$-prox-bounded;
		\item[\textup{\textbf{A3}}] $\varphi(u)$ admits at least one local solution.
	\end{enumerate}
\end{assumption}

Under Assumption~\ref{ass:basic}, both the gradient step on $\ell$ and the proximal step associated with $g$ are well defined. We may therefore introduce the forward-backward splitting (FBS) mapping associated with \eqref{eq:composite-problem} as
\begin{equation}
	\bar u \in T_\gamma(u)
	:=
	\operatorname{prox}_{\gamma g}\bigl(u-\gamma\nabla \ell(u)\bigr),
	\label{eq:fbs-step}
\end{equation}
namely a gradient step on the smooth part $\ell$ followed by a proximal correction induced by $g$.
In most control applications considered in this paper, the nonsmooth term $g$ is in fact chosen as
an indicator function that encodes hard constraints on the control input.
When \(g\) is the indicator function, i.e., \(g(u)=\delta_{\mathcal U}(u)\), where \(\delta_{\mathcal U}(u)=0\) if \(u\in\mathcal U\) and \(\delta_{\mathcal U}(u)=+\infty\) otherwise, the proximal mapping reduces to the Euclidean projection, and \eqref{eq:fbs-step} becomes the classical projected-gradient step.
\vspace{-2mm}

Associated with the selected FBS point $\bar u$ is the forward-backward residual:
\begin{equation}
	R_\gamma(u)
	:=
	\gamma^{-1}(u-\bar u).
	\label{eq:residual}
\end{equation}
A point $u^*$ is called
\emph{$\gamma$-critical} if $0\in R_\gamma(u^*)$, equivalently $u^*\in T_\gamma(u^*)$.
A key property of FBS is the following sufficient decrease estimate.
\vspace{-2mm}

\begin{lemma}[Sufficient decrease]
	\label{lem:sufficient-decrease}
	For any \(u\) and any \(\bar u \in T_\gamma(u)\), when 
	\(\gamma \in (0,\min\{\gamma_g,1/L_\ell\})\),
	\begin{equation}
		\varphi(u)-\varphi(\bar u)
		\ge
		\frac{\gamma(1-\gamma L_\ell)}{2}\|R_\gamma(u)\|^2 .
		\label{eq:sufdec}
	\end{equation}
\end{lemma}
\subsection{Newton-Type Method}
\vspace{-2mm}
Although FBS provides a robust descent mechanism, its convergence speed is inherently limited. To further accelerate local convergence, a Newton-type acceleration is incorporated.
\vspace{-2mm}

Under Assumption~\ref{ass:basic}, for any $\gamma\in(0,\gamma_g)$, every $\gamma$-critical point is
a stationary point of \eqref{eq:composite-problem}~\cite[Prob~3.5]{themelis2018forward}. Thus, the solver can target the fixed-point residual equation:
\begin{equation}
	0 \in R_\gamma(u),
	\label{eq:stationarity-root}
\end{equation}
which motivates Newton iterations of the form:
\begin{equation}
	u^{k+1}=u^k-H_k R_\gamma(u^k),
	\label{eq:newton-type}
\end{equation}
where \(H_k\) serves as an inverse-Jacobian approximation.
The above iteration relies on local second-order regularity of the underlying functions. We therefore impose the following assumptions.
\vspace{-2mm}
\begin{assumption}[Local regularity]
	\label{ass:local}
	Let $u^*$ be a $\gamma$-critical point of \eqref{eq:composite-problem}. We assume:
	\vspace{-2mm}
	\begin{enumerate}
		\item[\textup{\textbf{B1}}] $\nabla_{uu}^2 \ell$ and $\nabla_{u\theta}^2 \ell$ exist and are continuous in a neighborhood of $u^*$;
		\item[\textup{\textbf{B2}}] $g$ is prox-regular and $C^2$-partly smooth at $u^*$ with respect to a locally identified active manifold, and the strict complementarity condition holds for $-\nabla_u \ell$.
		\item[\textup{\textbf{B3}}] $u^*$ is a strong local minimum of \eqref{eq:composite-problem}.
	\end{enumerate}
\end{assumption}

For the definitions of prox-regular, see, e.g., \cite[Defs.~13.6]{rockafellar1998variational}.  
Assumptions \textbf{B1}--\textbf{B2} ensure the local existence of \(\nabla_u R_\gamma\), so that the Newton-type step in \eqref{eq:newton-type} is well defined. In practice, \(H_k\) is often constructed by a quasi-Newton update such as limited-memory BFGS (L-BFGS), which is attractive due to its low memory footprint. However, although Newton-type methods typically enjoy fast local convergence, they cannot in general guarantee convergence when the iterate is still far from the solution.

\subsection{Forward-Backward Envelope}
\vspace{-2mm}
To ensure global convergence of the Newton-type acceleration, we introduce the forward-backward envelope (FBE), defined as
\begin{equation}
	\varphi_\gamma(u)
	:=
	\min_{w}
	\Bigl\{
	\ell(u)
	+
	\langle\nabla \ell(u),\,w-u\rangle
	+
	g(w)
	+
	\frac{1}{2\gamma}\|w-u\|^2
	\Bigr\}.
	\label{eq:fbe-def}
\end{equation}
It can be viewed as a min-max formulation built upon a linearization and a quadratic approximation of the smooth part. Its basic properties are as follows:
\setcounter{lemma}{1}
\vspace{-2mm}
\begin{lemma}[Basic properties of the FBE]
	For any $\gamma \in (0,\gamma_g)$, $\varphi_\gamma$ is real-valued and continuous, and the following statements hold:
	\vspace{-2mm}
	\begin{enumerate}
		\item[\textup{(i)}] $\varphi_\gamma(\bar u) \le \varphi(\bar u)$.
		\item[\textup{(ii)}] $\varphi(\bar u) \le \varphi_\gamma(u) - \frac{\gamma(1-\gamma L_\ell)}{2}\|R_\gamma(u)\|^2$ for $\bar u \in T_\gamma(u)$.
		\item[\textup{(iii)}] $\arg\min \varphi=\arg\min \varphi_\gamma$ for $\gamma < 1/L_\ell$.
	\end{enumerate}
\end{lemma}

The continuity of \(\varphi_\gamma\) and the fact that it shares the same minimizers as \(\varphi\) make the FBE a suitable merit function for line search. Moreover, under Assumption~2, the FBE also admits useful local second-order structure. Define
\begin{equation}
	\scalebox{1.0}{$
		\begin{aligned}
			Q_\gamma := I - \gamma \nabla^2 \ell(u^*), \quad
			P_\gamma := J_u \operatorname{prox}_{\gamma g}
			\bigl(u^\star-\gamma\nabla\ell(u^\star)\bigr),\\
			\nabla_{u} R_\gamma := \frac{I - P_\gamma Q_\gamma}{\gamma}, \quad
			H_\gamma := \nabla^2 \varphi_\gamma(u^*) = Q_\gamma \, \nabla_{u} R_\gamma.
		\end{aligned}
		$}
	\label{eq:second}
\end{equation}
For $\gamma \in (0,1/L_\ell)$, $Q_\gamma$ and $H_\gamma$ are symmetric positive definite, $P_\gamma$ is symmetric positive semidefinite, and $\nabla_u R_\gamma$ is nonsingular~\cite[Thms.~4.10--4.11]{themelis2018forward}. These properties will be useful in the convergence analysis.

\subsection{Implicit differentiation of the solution map}
\vspace{-2mm}
In the backward pass, the objective is to compute the gradient of the outer loss in \eqref{eq:outer-loss} with respect to \(\theta\). By the chain rule:
\begin{equation}
	\nabla_\theta \mathcal L(\theta)
	=
	\nabla_u L(u^*)\,
	\frac{\partial u^*(\theta)}{\partial \theta},
	\label{eq:bp-chain}
\end{equation}
where the central quantity is the sensitivity \(\partial u^*(\theta)/\partial\theta\). Rather than invoking the implicit function theorem from the conventional KKT perspective, we proceed from the critical-point characterization \eqref{eq:stationarity-root}. Under Assumptions~\textbf{B1}-\textbf{B2}, the residual map \(R_{\gamma,\theta}(u)\) is continuously differentiable in a neighborhood of \((u^*,\theta)\), and under Assumption~\textbf{B3}, \(\nabla_u R_{\gamma,\theta}(u^*)\) is nonsingular. Therefore, by the implicit function theorem,
\begin{equation}
	\frac{\partial u^*(\theta)}{\partial\theta}
	=
	-\bigl(\nabla_u R_{\gamma,\theta}(u^*)\bigr)^{-1}
	\nabla_\theta R_{\gamma,\theta}(u^*).
	\label{eq:du-dtheta}
\end{equation}
The efficient evaluation of \eqref{eq:du-dtheta} and \eqref{eq:bp-chain} is the central issue which will be addressed in Section~3.2.
\vspace{-2mm}
\section{PANDA: A Matrix-free Solver for Differentiable NMPC}

Building on the above preliminaries, we next develop \texttt{PANDA}, i.e., a \emph{Proximal Averaged quasi-Newton with aDaptive linesearch Algorithm}, for efficient computations of both the forward and backward passes in differentiable optimization, as detailed below.
\subsection{Forward pass of PANDA}
\vspace{-2mm}
Before presenting the algorithm, we emphasize that most of the preceding properties hold under $\gamma\in(0,1/L_\ell)$, which is enforced through the following condition:
\begin{equation}
	\ell(\bar u)
	\le
	\ell(u)
	+
	\langle \nabla \ell(u),\, \bar u-u\rangle
	+
	\frac{\alpha}{2\gamma}
	\|\bar u-u\|^2,
	\label{eq:panda_upper_model}
\end{equation}
where $\alpha$ serves as the scaling factor in $\gamma=\alpha/L_\ell$.
The resulting forward pass is given in Algorithm~\ref{alg:panda}.

\begin{algorithm}
	\caption{Forward pass of PANDA}
	\label{alg:panda}
	\begin{algorithmic}[1]
		\Require $u^0\in\mathbb{R}^{N_p n_u}$, $\gamma_0\in(0,1/L_\ell)$, $\alpha,\beta\in(0,1)$, $N_\gamma>0$, $\varepsilon_r>0$, $\gamma_{max}>0$
		\State Initialize $c_\gamma\gets 0$, $\bar u^0\in T_{\gamma_0}(u^0)$ and $\Phi_0\gets \varphi_{\gamma_0}(u^0)$
		\For{$k=1,2,\ldots$}
		\State $\gamma_k\gets \gamma_{k-1}$
		\If{$c_\gamma\ge N_\gamma$ \textbf{and} $\|R_{\gamma_{k-1}}(u^{k-1})\|>\varepsilon_r$}
		\State Compute $\hat u^{k-1} \in T_{2\gamma_k}(u^{k-1})$
		\While{$(u^{k-1}, \hat u^{k-1}, 2\gamma_k)$ satisfies \eqref{eq:panda_upper_model}}
		\State $\gamma_k \gets 2\gamma_k,\quad c_\gamma \gets 0,\quad$ update $\hat u^{k-1}$
		\State \textbf{if} $\gamma_k > \gamma_{\max}$ \textbf{break}
		\EndWhile
		\EndIf
		\State Select a quasi-Newton direction $d^k$ such that $\|d^k\|\le D\|\bar u^{k-1}-u^{k-1}\|$ and set $\tau_k\gets 1$
		\State $u^k\gets (1-\tau_k)\bar u^{k-1}+\tau_k(u^{k-1}+d^k)$
		\State Compute $\bar u^k\in T_{\gamma_k}(u^k)$ and $\Phi_k\gets \varphi_{\gamma_k}(u^k)$
		\If{$(u^k,\bar u^k,\gamma_k)$ violates \eqref{eq:panda_upper_model}}
		\State $\gamma_k\gets \gamma_k/2,\quad c_\gamma\gets 0$
		\State Go back to line 11
		\EndIf
		\If{$\Phi_k>\Phi_{k-1}-\beta\frac{\gamma_{k-1}(1-\alpha)}{2}\|R_{\gamma_{k-1}}(u^{k-1})\|^2$}
		\State $\tau_k\gets \tau_k/2$
		\State Go back to line 12
		\EndIf
		\State $c_\gamma\gets c_\gamma+1$
		\EndFor
	\end{algorithmic}
\end{algorithm}
\vspace{-2mm}
In Algorithm~\ref{alg:panda}, the method first attempts to enlarge the current stepsize \(\gamma_k\) when the counter condition is met and the residual remains above the tolerance (Steps~4--10). Then, a trial point is generated through the averaged update (Steps~11--12). Steps~13--21 then combine the upper-model condition and the FBE line search to adaptively reduce \(\gamma_k\) and backtrack \(\tau_k\).
Overall, the algorithm combines the tools introduced above into a unified forward scheme. Its main distinction from prior methods lies in the adaptive stepsize enlargement.
\vspace{-2mm}

\begin{remark}[Local Lipschitz property]
	\label{rem:panda-gamma}
	In practical NMPC problems, the local Lipschitz continuity of $\nabla \ell$ is a very mild assumption. Compared with a global Lipschitz constant, a local one better captures the varying curvature of $\ell$ over different regions of the decision space. In particular, when the iterates move from a ``steeper'' region to a ``flatter'' one, the monotone stepsize reduction strategy used in prior works may leave the algorithm with an overly conservative estimate of the local smoothness. The stepsize enlargement mechanism in PANDA is introduced precisely to reduce this mismatch and accelerate the iterations in such situations.
\end{remark}

\begin{figure}[t]
	\centering
	\includegraphics[
	width=0.48\textwidth,
	trim=10pt 16pt 28pt 50pt,
	clip
	]{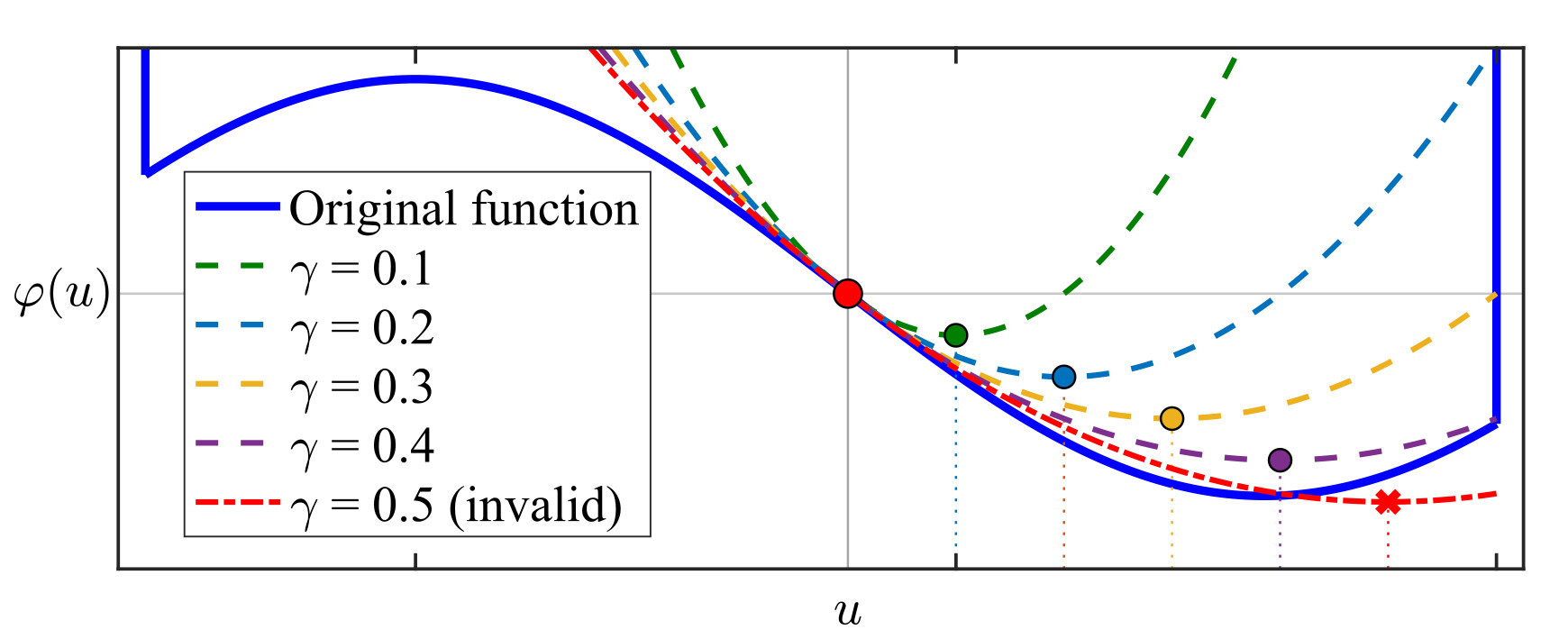}
	\caption{Illustration of the effect of the stepsize $\gamma$ on the FBS update for the example \(f(x)=\cos x\) under box constraints.}
	\label{fig:gamma}
\end{figure}

Figure~\ref{fig:gamma} illustrates the FBS updates induced by different stepsizes. It can be seen that, within a reasonable admissible range, a larger stepsize can yield a larger cost decrease than the more conservative update produced by a smaller one.
This observation motivates the stepsize enlargement
mechanism. We next examine its effect on
the convergence properties of the proposed algorithm.

\begin{lemma}[Stabilization of the adaptive stepsize]
	\label{lem:panda-finite-gamma}
	Suppose that Assumptions~1 and~2 hold. Then the following statements hold.
	\vspace{-2mm}
	\begin{enumerate}
		\item For each outer iteration \(k\), the backtracking loop on \(\gamma_k\) and \(\tau_k\) terminate finitely.
		\item If, in addition, \(\varphi\) is level bounded, then \((\gamma_k)\) is asymptotically constant, namely, there exist \(\bar\gamma>0\) and \(k_0\in\mathbb N\) such that
		$
		\gamma_k=\bar\gamma, \forall k\ge k_0.
		$
	\end{enumerate}
\end{lemma}

\vspace{-4mm}
\begin{pf}
	For a fixed outer iteration \(k\), the trial point generated at Step~10 satisfies
	$
	\|u^k-\bar u^{k-1}\|
	\le
	(1+D)\|u^{k-1}-\bar u^{k-1}\|.
	$
	Therefore, all trial points generated within the $k$-th outer iteration remain in a bounded neighborhood, and by \cite[Lemma~5.1]{de2022proximal}, $\bar u^{k-1}$ is also bounded. Since \(\nabla \ell\) is locally Lipschitz continuous on the above bounded set, there exists \(L_\ell^{\mathrm{loc}}>0\) such that the upper-model condition \eqref{eq:panda_upper_model} holds whenever \(\gamma_k\le \alpha/L_\ell^{\mathrm{loc}}\). Hence, the loop in Steps~11-16 terminates finitely.
	For the loop in Steps~12-20, finite termination follows directly from Lemma~2.2(i-ii).
	Moreover, by the acceptance condition at Step~18,
	$
	\Phi_k
	\le
	\Phi_{k-1}
	-
	\beta\frac{\gamma_{k-1}(1-\alpha)}{2}\|R_{\gamma_{k-1}}(u^{k-1})\|^2.
	$
	Telescoping this inequality, we obtain
	$
	R_{\gamma_k}(u^k)\to0.
	$
	Hence, the enlargement trigger condition eventually fails, so no further enlargement is activated after finitely many iterations.
	Finally, if \(\varphi\) is level bounded, then \((\bar u^k)\) is bounded since \(\varphi(\bar u^k)\le \Phi_0\). Together with the boundedness argument above, this implies that all sufficiently late trial points and iterates are contained in a common bounded set \(\Omega\). Hence, there exists \(L_{\Omega}^{\mathrm{loc}}>0\) such that \eqref{eq:panda_upper_model} holds on \(\Omega\) whenever \(\gamma_k\le \alpha/L_{\Omega}^{\mathrm{loc}}\). Therefore, \(\gamma_k\) cannot be reduced indefinitely and eventually becomes constant.
\end{pf}
\begin{theorem}[Convergence of PANDA]
	\label{thm:panda-convergence}
	Let \((u^k)\) be the sequence generated by PANDA. Suppose that Assumptions~1 and~2 hold and that \(\varphi\) is level bounded. Then the following statements hold.
	\vspace{-2mm}
	\begin{enumerate}
		\item \((u^k)\) converges \(R\)-linearly to \(u^\star\).
		\item If, in addition, the Dennis--Mor\'e condition holds, then \((u^k)\) converges superlinearly to \(u^\star\).
	\end{enumerate}
\end{theorem}

\vspace{-4mm}
\begin{pf}
	By Lemma~\ref{lem:panda-finite-gamma}, the stepsize sequence \((\gamma_k)\) becomes constant after finitely many iterations; denote its limiting value by \(\bar\gamma\). Therefore, once the objective satisfies the Kurdyka–\L{}ojasiewicz (KL) property at the limit point, the linear and superlinear convergence analysis in \cite[Thms.~5.8--5.10]{themelis2018forward} applies.
	We next show that strong local minimality implies the KL inequality with exponent \(1/2\). Indeed, by \eqref{eq:second}, one has
	$
	H_{\bar\gamma}\succ0.
	$
	Hence, for \(u\) close to \(u^\star\), letting \(d:=u-u^\star\), the expansions of \(\varphi_{\bar\gamma}\) and \(\nabla\varphi_{\bar\gamma}\) at \(u^\star\) yield \(\varphi_{\bar\gamma}(u)-\varphi_{\bar\gamma}(u^\star)=\frac12 d^\top H_{\bar\gamma}d+o(\|d\|^2)\) and \(\nabla\varphi_{\bar\gamma}(u)=H_{\bar\gamma}d+o(\|d\|)\).
	Therefore,
	$
	\frac{\bigl(\varphi_{\bar\gamma}(u)-\varphi_{\bar\gamma}(u^\star)\bigr)^{1/2}}
	{\|\nabla\varphi_{\bar\gamma}(u)\|}
	\le
	\frac{1}{\sqrt{2\,\lambda_{\min}(H_{\bar\gamma})}}
	$
	for all \(u\) sufficiently close to \(u^\star\). This proves that \(\varphi_{\bar\gamma}\) satisfies the KL inequality at \(u^\star\) with desingularizing function \(\psi(s)=\rho s^{1/2}\), where
	$
	\rho=\frac{\sqrt{2}}{\sqrt{\lambda_{\min}(H_{\bar\gamma})}}.
	$
	\hfill\(\square\)
\end{pf}
\begin{theorem}[Effect of $\gamma$ on the asymptotic rate]
	\label{thm:panda-gamma-nonmonotone}
	Under the same assumptions as in Theorem~3.3, the derived upper bound on the asymptotic linear convergence constant depends non-monotonically on $\gamma$. The bound deteriorates as $\gamma\to 1/L^{\mathrm{loc}}_\ell$; it also deteriorates as $\gamma\to0$ when the active manifold is positive-dimensional.
\end{theorem}
\vspace{-4mm}
\begin{pf}
	Define
	$
	B_k:=\sum_{i\ge k}\|r^i\|.
	$
	By \cite[Thm.~5.9]{themelis2018forward}, the \(R\)-linear convergence of \((u^k)\) is equivalent to the \(Q\)-linear convergence of \((B_k)\). Moreover, using the KL inequality together with Lemma~2.2 as in that proof, one obtains
	$
	B_{k+1}\le\Bigl(1-\frac{1}{C_\gamma}\Bigr)B_k
	$
	with
	$
	C_\gamma:=\frac{\rho^2(1+\gamma L_\ell^{\mathrm{loc}})^2}{\sigma}.
	$
	Therefore, it suffices to examine how \(C_\gamma\) depends on \(\gamma\).
	By further expanding \eqref{eq:second}, one has
	$
	H_\gamma=\frac1\gamma Q_\gamma^{1/2}(I-Q_\gamma^{1/2}P_\gamma Q_\gamma^{1/2})Q_\gamma^{1/2}.
	$
	As \(\gamma\to 1/L_\ell^{\mathrm{loc}}\), the smallest eigenvalue of \(Q_\gamma\) tends to zero through the factor \(1-\gamma L_\ell^{\mathrm{loc}}\). Since \(P_\gamma\) remains locally bounded under \textbf{B2}, the \(I-Q_\gamma^{1/2}P_\gamma Q_\gamma^{1/2}\) term stays bounded. Hence \(\lambda_{\min}(H_\gamma)\to0\), which implies \(\rho\to+\infty\) and thus \(C_\gamma\to+\infty\).
	On the other hand, as \(\gamma\to0\), we may rewrite
	$
	H_\gamma=Q_\gamma\bigl((I-P_\gamma)/\gamma+P_\gamma\nabla^2\ell(u^\star)\bigr).
	$
	If the identified active manifold is positive-dimensional, then \(P_\gamma\) admits the local representation
	$
	P_\gamma=\Pi_S(I+\gamma M)^{-1}\Pi_S
	$
	with \(S\neq\{0\}\). Hence, along directions in \(S\), the eigenvalues of \((I-P_\gamma)/\gamma\) converge to finite constants, so \(\lambda_{\min}(H_\gamma)\) remains bounded as \(\gamma\to0\). Therefore, the dominant term in \(C_\gamma\) is \(\sigma^{-1}\). Since \(\sigma\propto\gamma\), it follows that \(C_\gamma\to+\infty\) as \(\gamma\to0\).
	\hfill\(\square\)
\end{pf}
\vspace{-4mm}
With the above considerations, in typical NMPC problems, the optimal solution is usually not fully fixed by the active constraints, but still retains feasible tangential directions along the active manifold. In such cases, the stepsize enlargement mechanism in our algorithm, together with the safeguard parameter \(\alpha\), helps keep \(\gamma\) within a favorable range for the local asymptotic rate.

\vspace{-2mm}
The above discussion concerns the local convergence behavior near the optimum. Away from the optimum, \(\gamma\) also has several concrete algorithmic effects: a larger admissible \(\gamma\) makes the FBE acceptance condition in Step~18 more restrictive, which helps reject poor quasi-Newton trial steps and improves line-search robustness; moreover, \(\gamma\) enters the local quasi-Newton model through the Jacobian of the proximal mapping at the FBS point, and when \(g\) is an indicator function of constraints, this may help the search direction detect the active boundary earlier. On the other hand, changing \(\gamma\) is not cost-free: when L-BFGS is used, the accumulated secant information should be cleared after each update of \(\gamma\), making the choice of \(N_\gamma\) important in practice.

\subsection{Backward pass of PANDA}
\vspace{-2mm}
In this subsection, we introduce an iterative method for computing the parameter
gradient in the backward pass, while also exploiting the computational
simplifications induced by common constraint structures. 

Specifically, in view of \eqref{eq:bp-chain} and \eqref{eq:du-dtheta}, we adopt an adjoint-based formulation instead of explicitly forming the Jacobian \(\partial u^*(\theta)/\partial \theta\).
Let $\lambda$ be the solution of
\begin{equation}
	\bigl(\nabla_u R_{\gamma,\theta}(u^*)\bigr)^\top \lambda
	=
	\nabla_u L(u^*)^\top,
	\label{eq:bp-adjoint-system}
\end{equation}
then the outer gradient can be written as
\begin{equation}
	\nabla_\theta \mathcal L(\theta)
	=
	-\lambda^\top \nabla_\theta R_{\gamma,\theta}(u^*).
	\label{eq:bp-adjoint-grad}
\end{equation}
Therefore, the backward pass reduces to two steps: solve the linear system \eqref{eq:bp-adjoint-system} and compute the product \eqref{eq:bp-adjoint-grad}.
With the notation of \eqref{eq:second}, and define
$
S_\gamma := J_\theta\operatorname{prox}_{\gamma g}\bigl(u^\star-\gamma\nabla\ell(u^\star)\bigr),
$ the chain rule gives
\vspace{-2mm}
\begin{equation}
	\begin{aligned}
		\nabla_u R_{\gamma,\theta}(u^\star)
		&=
		\frac{1}{\gamma}(I-P_\gamma)+P_\gamma\nabla_{uu}^2\ell(u^\star,\theta), \\
		\nabla_\theta R_{\gamma,\theta}(u^\star)
		&=
		P_\gamma\nabla_{u\theta}^2\ell(u^\star,\theta)
		-\frac{1}{\gamma}S_\gamma .
	\end{aligned}
	\label{eq:bp-JR-JthetaR}
\end{equation}
To enable the matrix-free procedure described next, we assume that \(g\) is either separable or otherwise structured so that the operator evaluations associated with \(P_\gamma\) and \(S_\gamma\) can be implemented efficiently.
\vspace{-2mm}

\textbf{Step 1: Solve the linear system.}
For a general nonsmooth term $g$, the matrix
$\nabla_u R_{\gamma,\theta}(u^\star)$ is in general nonsymmetric, since $P_\gamma$ does
not commute with $\nabla_{uu}^2\ell(u^\star,\theta)$. In this case,
\eqref{eq:bp-adjoint-system} is solved iteratively by GMRES, which can be
implemented in a fully matrix-free manner, because for any vector $z$
\begin{equation}
	\scalebox{0.91}{$
		\displaystyle
		\bigl(\nabla_u R_{\gamma,\theta}(u^\star)\bigr)^\top z
		=
		\frac{1}{\gamma}(I-P_\gamma^\top)z
		+
		\nabla_{uu}^2\ell(u^\star,\theta)\bigl(P_\gamma^\top z\bigr),
		$}
	\label{eq:bp-JRT-z}
\end{equation}
where the second term is a Hessian-Vector product computed via the automatic-differentiation computational graph, so that no explicit matrix is required.

\vspace{-2mm}
A particularly important case is when $g=\delta_{\mathcal U}$ is the indicator
function of box constraints. Then $P_\gamma$ becomes a diagonal $0$-$1$ mask,
$
P_\gamma=\operatorname{diag}(p_1,\dots,p_{N_pn_u}),
p_j\in\{0,1\},
$
which induces the partition
$
\mathcal F:=\{j:p_j=1\},
\mathcal A:=\{j:p_j=0\},
$
corresponding to the free and active variables, respectively. After reordering
the variables, one has
\vspace{-2mm}
\begin{equation}
P_\gamma=
\begin{bmatrix}
	I & 0\\
	0 & 0
\end{bmatrix},
\qquad
\nabla_{uu}^2\ell(u^*,\theta)=
\begin{bmatrix}
	H_{\mathcal {FF}} & H_{\mathcal {FA}}\\
	H_{\mathcal {AF}} & H_{\mathcal {AA}}
\end{bmatrix}.
\end{equation}
Writing \(\lambda=[\lambda_{\mathcal F}^\top,\lambda_{\mathcal A}^\top]^\top\) and partitioning \(\nabla_u L(u^*)\) accordingly, the adjoint system \eqref{eq:bp-adjoint-system} reduces to
\vspace{-2mm}
\begin{equation}
	\begin{bmatrix}
		H_{\mathcal {FF}} & 0\\
		H_{\mathcal {AF}} & \frac{1}{\gamma}I
	\end{bmatrix}
	\begin{bmatrix}
		\lambda_{\mathcal F}\\
		\lambda_{\mathcal A}
	\end{bmatrix}
	=
	\begin{bmatrix}
		\nabla_{u_{\mathcal F}}L(u^*)^\top\\
		\nabla_{u_{\mathcal A}}L(u^*)^\top
	\end{bmatrix}.
	\label{eq:bp-adjoint-block}
\end{equation}
In particular, the first block yields the reduced symmetric system:
\begin{equation}
	H_{\mathcal {FF}}\lambda_{\mathcal F}
	=
	\nabla_{u_{\mathcal F}}L(u^*)^\top,
	\label{eq:bp-adjoint-free}
\end{equation}
which can be solved efficiently by MINRES using Hessian-Vector products only.
The active component \(\lambda_{\mathcal A}\) is then recovered
from the second block as
\vspace{-2mm}
\begin{equation}
	\lambda_{\mathcal A}
	=
	\gamma\Bigl(
	\nabla_{u_{\mathcal A}}L(u^*)^\top
	-
	H_{\mathcal {AF}}\lambda_{\mathcal F}
	\Bigr).
	\label{eq:bp-lambdaA-recovery}
\end{equation}
\textbf{Step 2: Compute the Vector-Jacobian product.}
Once the adjoint variable is obtained, the remaining term in
\eqref{eq:bp-adjoint-grad} is
\begin{equation}
	\nabla_\theta \mathcal L(\theta)
	=
	-\lambda^\top P_\gamma\nabla_{u\theta}^2\ell(u^\star,\theta)
	+\frac{1}{\gamma}\lambda^\top S_\gamma .
	\label{eq:bp-second-term}
\end{equation}
This term can be evaluated by Vector-Jacobian products, without explicitly forming the corresponding full Jacobian matrices.
\vspace{-2mm}

The overall backward pass of \texttt{PANDA} is summarized in
Algorithm~\ref{alg:backward-pass}.

\begin{algorithm}[H]
	\caption{Backward pass of PANDA}
	\label{alg:backward-pass}
	\begin{algorithmic}[1]
		\Statex \textbf{Inputs:} local solution $u^*$ from Algorithm~\ref{alg:panda}, parameter $\theta$, stepsize $\gamma$, stored computation graph
		\Statex \textbf{Output:} outer gradient $\nabla_\theta \mathcal L(\theta)$
		
		\State Evaluate $P=J_u\operatorname{prox}_{\gamma g}(u^*-\gamma\nabla_u \ell(u^*,\theta))$
		
		\If{$g=\delta_{\mathcal U}$ with $\mathcal U$ being box constraints}
		\State Identify the free set $\mathcal F$ and the active set $\mathcal A$
		\State Solve \eqref{eq:bp-adjoint-free} by MINRES to obtain $\lambda_{\mathcal F}$
		\State Evaluate \eqref{eq:bp-lambdaA-recovery} to obtain $\lambda_{\mathcal A}$
		\Else
		\State Solve \eqref{eq:bp-adjoint-system} by GMRES to obtain $\lambda$
		\EndIf
		
		\State Compute $\nabla_\theta \mathcal L(\theta)=-\lambda^\top \nabla_\theta R_{\gamma,\theta}(u^*)$ by \eqref{eq:bp-second-term}
		\State Return $\nabla_\theta \mathcal L(\theta)$
	\end{algorithmic}
\end{algorithm}

\section{Experiments}
\vspace{-2mm}
In this section, to evaluate the proposed method, we consider an imitation learning task built upon a nonconvex trajectory optimization problem for a mobile robot carrying a trailer, and assess its performance from three complementary perspectives: computational time, memory consumption, and learning performance.
\subsection{Problem setup}
\vspace{-2mm}
For the trailer system, we define the state as \(x=[p_x,p_y,\phi]^\top\), where \(p_x\) and \(p_y\) denote the trailer position and \(\phi\) denotes its heading angle, and the control input as \(u=[u_x,u_y]^\top\), where \(u_x\) and \(u_y\) are the planar velocity commands. Using forward Euler discretization, the trailer kinematics can be written as
\begin{equation}
	x_{k+1}
	=
	x_k
	+
	\Delta t
	\begin{bmatrix}
		u_{x,k} + L\sin\phi_k \dot\phi_k \\
		u_{y,k} - L\cos\phi_k \dot\phi_k \\
		\dfrac{u_{y,k}\cos\phi_k-u_{x,k}\sin\phi_k}{L}
	\end{bmatrix}.
\end{equation}
We consider a point-to-point planning problem from an initial state \(x_0\) to a target state \(x_{\mathrm{ref}}\), while avoiding a circular obstacle centered at \(c\in\mathbb{R}^2\). The resulting optimal control problem is formulated as
\vspace{-2mm}
\begin{equation}
	\scalebox{0.95}{$
		\begin{aligned}
			\underset{\{x_k,u_k\}}{\text{min}}\quad
			& \sum_{k=0}^{N-1} (
			\|x_k-x_{\mathrm{ref}}\|_Q^2
			+ \|u_k\|_R^2 \\
			&
			+ \tfrac{1}{2}\eta_{\mathrm{obs}} \max\!\bigl(h(x_k),0\bigr)^2
			)
			+ \|x_N-x_{\mathrm{ref}}\|_W^2. \\
			\text{s.t. }\quad
			& x_0 = \bar x_0, \\
			& x_{k+1} = f(x_k,u_k), \\
			& -u^{\max} \le u_k \le u^{\max}, \quad k=0,\dots,N-1.
		\end{aligned}
		$}
\end{equation}
where \(Q=\mathrm{diag}(q_1,q_2,q_3)\), \(R=\mathrm{diag}(r_1,r_2)\), and \(W=\mathrm{diag}(w_1,w_2,w_3)\). The obstacle penalty is defined as
$
h(x_k)=1-\frac{\|p_k-c\|_2^2}{r_{\mathrm{safe}}^2}, p_k=[p_{x,k},p_{y,k}]^\top,
$
where \(r_{\mathrm{safe}}\) denotes the safety distance. 

For the imitation learning task, the learnable parameter is defined as
\(
\theta=[q_1,q_2,q_3,r_1,r_2,w_1,w_2,w_3,\eta_{\mathrm{obs}},r_{\mathrm{safe}}]^\top
\).
Given a demonstration pair \((x_{0:N}^{\mathrm{tar}},u_{0:N-1}^{\mathrm{tar}})\) generated by the teacher parameter \(\theta_{\mathrm{teacher}}\), the goal is to update an initial parameter vector \(\theta_{\mathrm{student}}\) so that the optimal trajectory induced by \(\theta\) progressively matches the demonstrated one. The outer imitation loss penalizes the deviations of the optimized state and control trajectories from their target counterparts, weighted by \(Q_{\mathrm{IL}}=\mathrm{diag}(10,10,2)\) and \(R_{\mathrm{IL}}=\mathrm{diag}(2,2)\), respectively. Unless otherwise specified, we set  \(N=40\) and \(\Delta t=0.1\). The teacher parameter is chosen as
$
\theta_{\mathrm{teacher}}
=
[0.6,\,0.6,\,0.03,\,0.02,\,0.04,\,4.0,\,4.0,\,0.8,\,100,\,0.7]^\top,
$
and \(\theta_{\mathrm{student}}\) is initialized by adding random perturbations. Training runs for \(50\) epochs with learning rate \(2 \times 10^{-4}\). All experiments are performed on an Intel(R) Core(TM) i5-12600KF CPU at 3.7\,GHz with 32\,GB RAM.

\subsection{Baselines}
\vspace{-2mm}
We compare \texttt{PANDA} with three representative differentiable solvers: SafePDP \cite{jin2021safe}, which computes gradients from the auxiliary-system perspective; CasADi \cite{andersson2018casadi,andersson2018sensitivity}, which employs a primal-dual active-set method in the backward pass; and acados \cite{verschueren2022acados,frey2025differentiable}, which relies on an interior-point sensitivity framework built on HPIPM. In addition, PANOC+\cite{de2022proximal} is included only in the forward pass as an ablation.

\subsection{Validation of adaptive stepsize}
\vspace{-2mm}
\begin{figure}[t] 
	\centering 
	\includegraphics[width=0.48\textwidth,
	trim=1cm 0.2cm 5cm 3.5cm,
	clip
	]{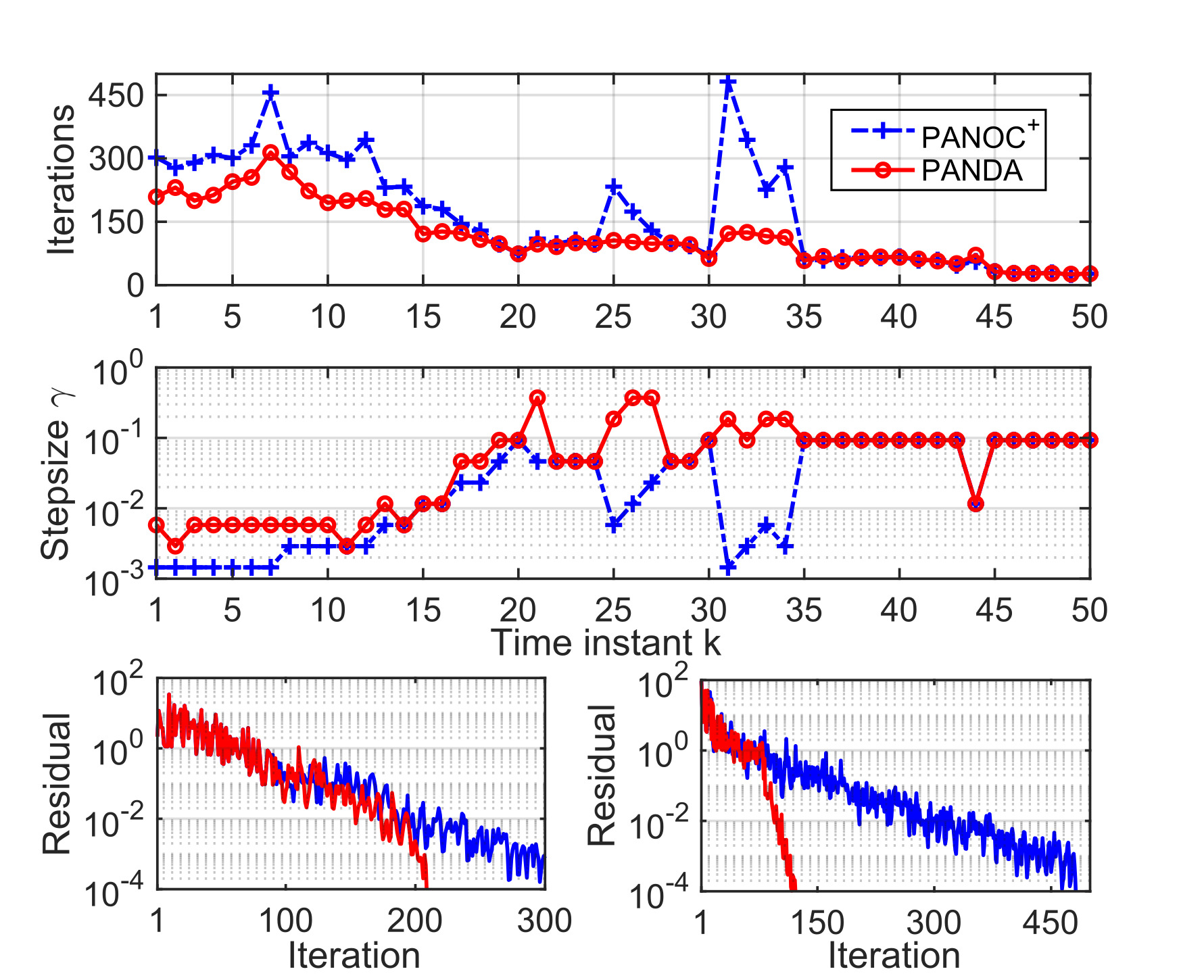} 
	\caption{Comparison between \texttt{PANDA} and PANOC+ on the forward optimization problem, showing iterations and final stepsizes at each time instant, together with residual curves at the initial instant \(k=1\) and the spike instant \(k=31\).}
	\label{fig:panda_panoc_time} 
\end{figure}

With the setting \(N_{\gamma}=80\) and \(\epsilon=10^{-3}\), we solve the above NMPC problem over 50 time instants using \texttt{PANDA} and PANOC+. The results are shown in Fig.~\ref{fig:panda_panoc_time}. It can be seen that \texttt{PANDA} requires fewer iterations than PANOC+, with the improvement being especially pronounced at the spike points around time instants \(31\)--\(34\). The final stepsizes indicate that, at these instants, PANOC+ selects an overly small stepsize compared with the neighboring time instants, e.g., \(\gamma\approx 0.002\) versus \(\gamma\approx 0.09\), resulting in a significant increase in iteration count. By contrast, \texttt{PANDA} keeps \(\gamma\) in a more effective range through adaptive enlargement, thereby avoiding such spikes. The residual curves at the initial and spike instants show that \texttt{PANDA} exhibits a more stable linear convergence behavior. These observations are consistent with the discussion in Section~3.1, where adaptively enlarging the admissible stepsize helps avoid overly conservative updates and improves convergence efficiency.

\subsection{Overall runtime in imitation learning}
\vspace{-2mm}
\begin{figure}[t] 
	\centering 
	\includegraphics[width=0.48\textwidth,
	trim=0.3cm 0.4cm 4cm 2cm,
	clip
	]{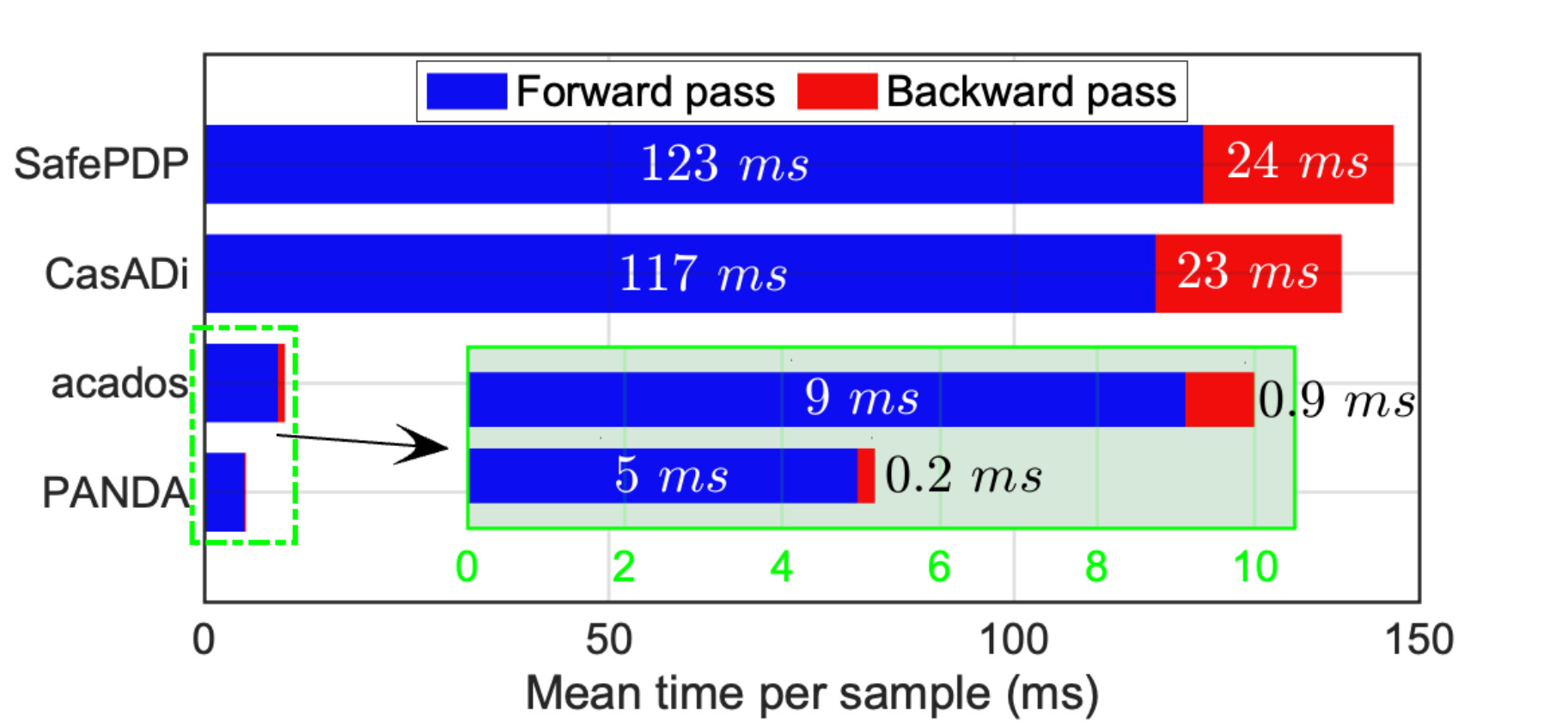} 
	\caption{Average forward pass, backward pass, and total runtime of different methods in the imitation learning task.} \label{fig:runtime_bar} 
\end{figure}

We record the average forward and backward computation times of each method over the imitation learning task, as summarized in Fig.~\ref{fig:runtime_bar}. \texttt{PANDA} outperforms SafePDP and CasADi, and runs faster than acados, achieving approximately 2× speedup in the forward pass, which comes from the combination of several efficient ingredients and the adaptive stepsize enlargement strategy, and 4× speedup in the backward pass, which benefits from the lightweight MINRES solve with an average of 14 iterations.

\subsection{Memory consumption analysis}
\vspace{-2mm}
We record the memory overhead of different methods for the same problem under different prediction horizons, as summarized in Table~\ref{tab:memory}. We exclude the memory used for problem construction and function initialization, and instead compare the buffer memory required for first- and second-order information during the solver iterations.

\begin{table}[h]
	\centering
	\caption{Memory overhead (MB) under different prediction horizons.}
	\label{tab:memory}
	\begin{tabular}{lcccc}
		\hline
		Method & $N=20$ & $N=40$ & $N=60$ & $N=80$ \\
		\hline
		SafePDP & 1.52 & 2.12 & 2.85 & 3.83 \\
		CasADi  & 1.38 & 2.07 & 2.74 & 4.36 \\
		acados  & 0.75 & 1.15 & 1.86 & 3.55 \\
		PANDA   & \textbf{0.28} & \textbf{0.39} & \textbf{0.52} & \textbf{0.65} \\
		\hline
	\end{tabular}
\end{table}

As shown in Table~\ref{tab:memory}, \texttt{PANDA} achieves the lowest memory overhead and its memory overhead grows approximately linearly, while the other methods increase much faster due to explicit matrix construction.
\vspace{-2mm}
\subsection{Learning performance}
\vspace{-2mm}

\begin{figure}[t]
	\centering
	\includegraphics[width=0.48\textwidth,
	trim=10pt 5cm 12pt 5cm,
	clip
	]{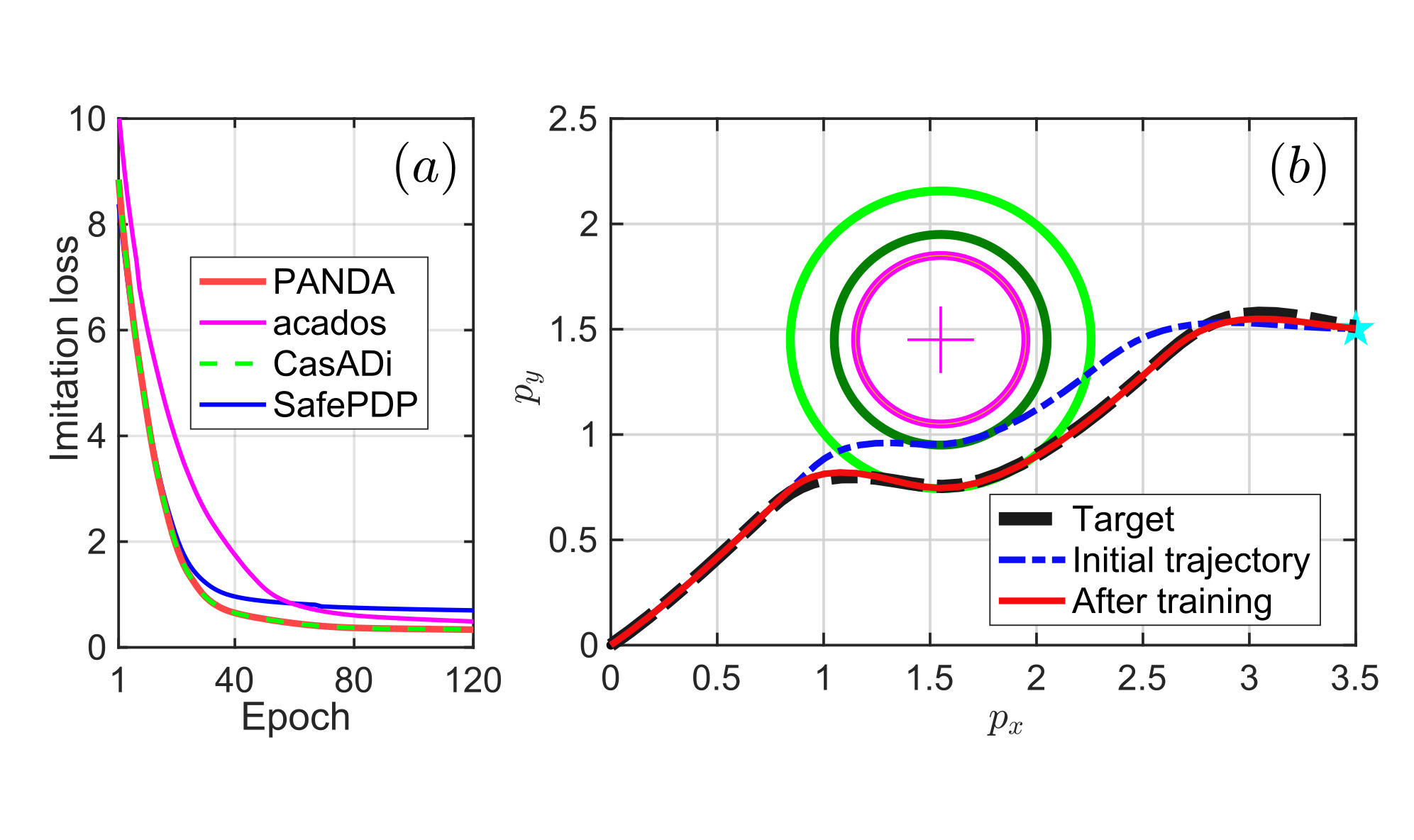}
	\caption{Learning performance of different methods: (a) training loss; (b) trajectory comparison before and after training.}
	\label{fig:loss_curve}
\end{figure}

Finally, the imitation learning performance of the compared methods in 120 epochs is shown in Fig.~\ref{fig:loss_curve}. The loss curves are broadly similar, except for acados, due to its inaccurate handling of non-positive-definite Hessian matrices, a limitation discussed in \cite[Rems.~1 and~2]{frey2025differentiable}. Moreover, as the parameter \(\theta\) is updated during training, the resulting optimal trajectory gradually approaches the demonstration style, which further verifies the effectiveness of the proposed method.
\vspace{-2mm}
\section{Conclusion}
\vspace{-2mm}
In this paper, we proposed \texttt{PANDA}, a fast and matrix-free solver for differentiable optimization. In the forward pass, \texttt{PANDA} combines several efficient techniques and accelerates convergence through an adaptive stepsize enlargement strategy. In the backward pass, it performs matrix-free computation by combining iterative method with Matrix-Vector product operators. The results on a practical nonconvex trailer optimal control problem, in terms of computation time, memory consumption, and imitation performance, demonstrate the effectiveness of the proposed method.
\vspace{-2mm}
\begin{ack}
\vspace{-4mm}
This work was supported in part by the National Natural Science
Foundation of China under Grant 62573209, in part by the Development
and Reform Commission Foundation of Jilin Province under
Grant 2023C034-3.
\end{ack}
\vspace{-4mm}
\bibliographystyle{plain}        
\bibliography{autosam}           

@article{chen2025gauss,
	title={A Gauss-Newton-Induced Structure-Exploiting Algorithm for Differentiable Optimal Control},
	author={Chen, Yuankun and Nie, Zifei and Gong, Xun and Hu, Yunfeng and Chen, Hong},
	journal={arXiv preprint arXiv:2512.19447},
	year={2025}
}

@article{adabag2025differentiable,
	title={Differentiable Model Predictive Control on the GPU},
	author={Adabag, Emre and Greiff, Marcus and Subosits, John and Lew, Thomas},
	journal={arXiv preprint arXiv:2510.06179},
	year={2025}
}

@article{huang2023E2E,
	title={Differentiable integrated motion prediction and planning with learnable cost function for autonomous driving},
	author={Huang, Zhiyu and Liu, Haochen and Wu, Jingda and Lv, Chen},
	journal={IEEE transactions on neural networks and learning systems},
	volume={35},
	number={11},
	pages={15222--15236},
	year={2023},
	publisher={IEEE}
}

@article{andersson2018casadi,
	title={{CasADi}: A Software Framework for Nonlinear Optimization and Optimal Control},
	author={Andersson, Joel A. E. and Gillis, Joris and Horn, Greg and Rawlings, James B. and Diehl, Moritz},
	journal={Mathematical Programming Computation},
	volume={11},
	number={1},
	pages={1--36},
	year={2019},
	doi={10.1007/s12532-018-0139-4},
	publisher={Springer}
}

@article{andersson2018sensitivity,
	title={Sensitivity analysis for nonlinear programming in CasADi},
	author={Andersson, Joel AE and Rawlings, James B},
	journal={IFAC-PapersOnLine},
	volume={51},
	number={20},
	pages={331--336},
	year={2018},
	publisher={Elsevier}
}

@article{verschueren2022acados,
	title={acados—a modular open-source framework for fast embedded optimal control},
	author={Verschueren, Robin and Frison, Gianluca and Kouzoupis, Dimitris and Frey, Jonathan and Duijkeren, Niels van and Zanelli, Andrea and Novoselnik, Branimir and Albin, Thivaharan and Quirynen, Rien and Diehl, Moritz},
	journal={Mathematical Programming Computation},
	volume={14},
	number={1},
	pages={147--183},
	year={2022},
	publisher={Springer}
}

@article{frey2025differentiable,
	title={Differentiable nonlinear model predictive control},
	author={Frey, Jonathan and  Baumgärtner, Katrin and Frison, Gianluca and Reinhardt, Dirk and Hoffmann, Jasper and Fichtner, Leonard and Gros, Sebastien and Diehl, Moritz},
	journal={arXiv preprint arXiv:2505.01353},
	year={2025}
}

@inproceedings{stella2017simple,
	title={A simple and efficient algorithm for nonlinear model predictive control},
	author={Stella, Lorenzo and Themelis, Andreas and Sopasakis, Pantelis and Patrinos, Panagiotis},
	booktitle={2017 IEEE 56th Annual Conference on Decision and Control (CDC)},
	pages={1939--1944},
	year={2017},
	organization={IEEE}
}

@article{themelis2018forward,
	title={Forward-backward envelope for the sum of two nonconvex functions: Further properties and nonmonotone linesearch algorithms},
	author={Themelis, Andreas and Stella, Lorenzo and Patrinos, Panagiotis},
	journal={SIAM Journal on Optimization},
	volume={28},
	number={3},
	pages={2274--2303},
	year={2018},
	publisher={SIAM}
}

@article{jin2021safe,
	title={Safe pontryagin differentiable programming},
	author={Jin, Wanxin and Mou, Shaoshuai and Pappas, George J},
	journal={Advances in Neural Information Processing Systems},
	volume={34},
	pages={16034--16050},
	year={2021}
}

@article{amos2018differentiable,
	title={Differentiable mpc for end-to-end planning and control},
	author={Amos, Brandon and Jimenez, Ivan and Sacks, Jacob and Boots, Byron and Kolter, J Zico},
	journal={Advances in neural information processing systems},
	volume={31},
	year={2018}
}

@book{rockafellar1998variational,
	title={Variational analysis},
	author={Rockafellar, R Tyrrell and Wets, Roger JB},
	year={1998},
	publisher={Springer}
}

@article{de2022proximal,
	title={Proximal gradient algorithms under local Lipschitz gradient continuity: A convergence and robustness analysis of PANOC},
	author={De Marchi, Alberto and Themelis, Andreas},
	journal={Journal of Optimization Theory and Applications},
	volume={194},
	number={3},
	pages={771--794},
	year={2022},
	publisher={Springer}
}



\appendix

\end{document}